\documentclass[11pt]{article}
\usepackage[T1]{fontenc}
\usepackage[utf8]{inputenc}
\usepackage{lmodern}
\usepackage{amsmath,amssymb,amsthm,mathtools,bm}
\usepackage{graphicx}
\usepackage{booktabs}
\usepackage{microtype}
\usepackage{geometry}
\usepackage{enumitem}
\usepackage[numbers,sort&compress]{natbib}
\usepackage[colorlinks=true,linkcolor=blue,citecolor=blue,urlcolor=blue]{hyperref}
\newtheorem{theorem}{Theorem}[section]
\newtheorem{proposition}[theorem]{Proposition}
\newtheorem{corollary}[theorem]{Corollary}

\newtheorem{remark}[theorem]{Remark}

\newcommand{\dd}{\,\mathrm{d}}
\newcommand{\Dx}{D_x}
\newcommand{\wpF}{\wp}

\title{Painlev\'e Integrability, Hamiltonian Structure\\and Exact Elliptic Reductions of a Generalized Nonlinear Wave Equation}
\author{Alvaro H. Salas\\
\small Department of Mathematics and Statistics\\
\small Universidad Nacional de Colombia, Sede Manizales, Colombia\\
\small FIZMAKO Research Group\\
\small \texttt{ahsalass@unal.edu.co} \quad ORCID: 0000-0001-9343-6062}
\date{}

\begin{document}
\maketitle

\begin{abstract}
We study the polynomial generalized Korteweg--de Vries family
\[
 u_t+P_m(u)u_x+\kappa u_{xxx}=0,\qquad \kappa\neq0,
\]
where $P_m$ is a real polynomial of degree $m\ge1$.  The purpose is not to generate isolated closed-form waves, but to identify the structural threshold at which pole-type Painlev\'e behavior and elliptic traveling-wave geometry are simultaneously lost.  A Weiss--Tabor--Carnevale dominant-balance calculation gives the universal principal exponent $p=-2/m$.  Hence only $m=1$ and $m=2$ can possess a principal Laurent branch with integer pole order; their resonance sets are respectively $\{-1,4,6\}$ and $\{-1,3,4\}$.  These two sectors are, after affine and Galilean transformations, the KdV and modified KdV equations, so their compatibility conditions and complete-integrability structures are inherited from the classical hierarchies.  In contrast, every degree $m\ge3$ has a fractional leading exponent and therefore fails the strong WTC pole criterion at the first step.

The whole polynomial family is nevertheless Hamiltonian with the Gardner bracket $D_x$.  Traveling-wave reduction yields a hyperelliptic curve $Y^2=R_{m+2}(U)$ whose generic genus is $\lfloor(m+1)/2\rfloor$.  Thus $m=1,2$ are precisely the nonlinear polynomial degrees for which the generic traveling-wave curve is elliptic, whereas the first Painlev\'e-obstructed case $m=3$ is also the first generic genus-two case.  For the elliptic sectors we derive a root-based Weierstrass representation valid for an arbitrary quartic first integral, characterize discriminant degenerations, and recover periodic, solitary, kink-type and rational limits.  A focusing mKdV periodic orbit is used for a reproducible numerical verification, with maximum pointwise error about $3.2\times10^{-12}$.  The results provide a concise bridge between singularity analysis, Hamiltonian form, algebraic-curve genus and exact nonlinear waves.
\end{abstract}



\section{Introduction}
The Korteweg--de Vries (KdV) and modified Korteweg--de Vries (mKdV) equations are canonical models in the theory of nonlinear dispersive waves and completely integrable systems \citep{korteweg1895,gardner1967,miura1968,ablowitz1981}.  Their importance is not tied to a single explicit soliton formula: they sit at the intersection of inverse scattering, Lax representations, Hamiltonian hierarchies, B\"acklund transformations, conservation laws and meromorphic singularity structure.  The Painlev\'e program of Ablowitz--Ramani--Segur and Weiss--Tabor--Carnevale (WTC) made this intersection especially transparent by connecting movable-singularity structure with integrable evolution equations \citep{ars1980a,ars1980b,weiss1983,weiss1983b,conte1999}.

A recurring generalization replaces the KdV characteristic speed by a nonlinear function of the field,
\begin{equation}
 u_t+P(u)u_x+\kappa u_{xxx}=0,\qquad \kappa\ne0,
 \label{eq:pgkdv}
\end{equation}
with $P$ polynomial.  Equation \eqref{eq:pgkdv} is Hamiltonian for every smooth $P$, and traveling waves are always reduced to an algebraic first integral.  These two facts alone, however, do not imply complete integrability.  In particular, an arbitrary polynomial characteristic speed generally changes both the movable-singularity structure of the PDE and the genus of the algebraic curve associated with a traveling wave.

Polynomial source terms added to KdV/mKdV have also been examined with Painlev\'e and exact-solution techniques \citep{kudryashov2014,kudryashov2016}, while broader KdV-type evolution classes have been organized by Lie-symmetry classification \citep{gungor2004}.  The family studied here is structurally different from a polynomial source perturbation: the polynomial remains inside the conservative characteristic speed $P(u)u_x$, so the $D_x$ Hamiltonian bracket survives for every degree.  This makes it possible to compare, within one fixed Hamiltonian class, the WTC pole threshold with the genus of the traveling-wave curve.

The aim of the present paper is to make that transition explicit.  Let
\begin{equation}
 P_m(u)=\sum_{j=0}^{m}a_j u^j,\qquad a_m\ne0,\qquad m\ge1.
 \label{eq:poly}
\end{equation}
The main structural observation is a sharp coincidence:
\begin{quote}
\emph{The nonlinear degrees $m=1,2$ are exactly the polynomial degrees for which the WTC dominant exponent is a negative integer and the generic traveling-wave curve is elliptic.  At $m=3$, the WTC exponent becomes fractional and the generic traveling-wave curve becomes genus two.}
\end{quote}
This is elementary to state, but it organizes several usually separate calculations into one classification principle.

The paper is deliberately conservative about the meaning of a failed Painlev\'e test.  Failure of the strong WTC pole criterion is used as an obstruction/diagnostic, not as a theorem excluding every possible notion of integrability.  Conversely, in the admissible $m=1,2$ sectors, complete integrability is established independently because affine and Galilean transformations reduce the equations to classical KdV or mKdV.

Our contributions are as follows.  First, we derive the degree-dependent WTC exponent and resonance polynomials.  Second, we exhibit a Hamiltonian formulation valid for the entire polynomial family and identify the affine normal forms of the integrable sectors.  Third, we connect polynomial degree to the genus of the traveling-wave curve.  Fourth, for the elliptic sectors we derive a root-based Weierstrass formula directly from an arbitrary quartic first integral and describe its degenerations.  Finally, we provide symbolic identities and a high-accuracy numerical verification that are included with the source package.

\section{Hamiltonian polynomial dispersive waves}
We work either on the real line with sufficiently rapid decay or on a periodic interval, so that integrations by parts generate no boundary terms.

\subsection{Conservation form and Hamiltonian functional}
Let $A'(u)=P_m(u)$ and $Q''(u)=P_m(u)$.  We choose
\begin{align}
 A(u)&=\sum_{j=0}^{m}\frac{a_j}{j+1}u^{j+1},\\
 Q(u)&=\sum_{j=0}^{m}\frac{a_j}{(j+1)(j+2)}u^{j+2}.
 \label{eq:AQ}
\end{align}
Then \eqref{eq:pgkdv} is the conservation law
\begin{equation}
 u_t+\partial_x\bigl(A(u)+\kappa u_{xx}\bigr)=0.
 \label{eq:conservation}
\end{equation}
In particular, the mass $M[u]=\int u\,\dd x$ is conserved.

\begin{proposition}[Canonical Hamiltonian form]
For every polynomial $P_m$, equation \eqref{eq:pgkdv} has the Hamiltonian representation
\begin{equation}
 u_t=\Dx\frac{\delta H}{\delta u},
 \qquad
 H[u]=\int\left(\frac{\kappa}{2}u_x^2-Q(u)\right)\dd x,
 \label{eq:ham}
\end{equation}
with Poisson operator $J_0=\Dx$.
\end{proposition}

\begin{proof}
The variational derivative is
\[
 \frac{\delta H}{\delta u}=-\kappa u_{xx}-Q'(u).
\]
Applying $D_x$ and using $Q''=P_m$ gives
\[
 \Dx\frac{\delta H}{\delta u}=-\kappa u_{xxx}-P_m(u)u_x,
\]
which is precisely \eqref{eq:pgkdv}.
\end{proof}

The quadratic quantity
\begin{equation}
 N[u]=\frac12\int u^2\,\dd x
 \label{eq:l2}
\end{equation}
is also conserved.  Indeed, $D_x(\delta N/\delta u)=u_x$, so $N$ is the momentum generating spatial translations for the Gardner bracket.  This observation is useful below: Hamiltonian structure persists even in the polynomial degrees that fail the Painlev\'e pole criterion, and therefore Hamiltonianity by itself does not distinguish the completely integrable sectors.

\section{Painlev\'e degree threshold}
We apply the WTC local singularity test near a noncharacteristic movable manifold
\[
 \phi(x,t)=0,\qquad \phi_x\ne0,
\]
and seek a principal behavior
\begin{equation}
 u\sim u_0\phi^p,\qquad p<0.
 \label{eq:dominant}
\end{equation}
The time derivative is less singular than the nonlinear convection/dispersion balance for the branches considered below.

\begin{theorem}[Universal dominant exponent]
Let $P_m$ have degree $m\ge1$.  Any dominant balance between $a_m u^m u_x$ and $\kappa u_{xxx}$ has exponent
\begin{equation}
 p=-\frac{2}{m}.
 \label{eq:p}
\end{equation}
Consequently, a principal Laurent branch with an integer pole order can occur only for $m=1$ or $m=2$.
\end{theorem}

\begin{proof}
From \eqref{eq:dominant},
\[
 u^m u_x=O\bigl(\phi^{(m+1)p-1}\bigr),
 \qquad
 u_{xxx}=O\bigl(\phi^{p-3}\bigr).
\]
Equality of the dominant exponents yields
\[
 (m+1)p-1=p-3,
\]
which gives \eqref{eq:p}.  The number $-2/m$ is a negative integer only for $m=1,2$.
\end{proof}

The leading coefficient follows after retaining the powers of $\phi_x$:
\begin{equation}
 a_m u_0^m+\kappa(p-1)(p-2)\phi_x^2=0.
 \label{eq:u0general}
\end{equation}
Thus for the two pole-admissible nonlinear degrees,
\begin{align}
 m=1:&\quad p=-2,\qquad u_0=-\frac{12\kappa\phi_x^2}{a_1},
 \label{eq:u0kdv}\\
 m=2:&\quad p=-1,\qquad u_0^2=-\frac{6\kappa\phi_x^2}{a_2}.
 \label{eq:u0mkdv}
\end{align}
Complex-valued local branches are permitted in Painlev\'e analysis, so the sign in \eqref{eq:u0mkdv} is not restrictive.

\subsection{Resonances}
Set
\begin{equation}
 u=u_0\phi^p+\epsilon u_r\phi^{p+r}
 \label{eq:respert}
\end{equation}
and retain terms linear in $\epsilon$ at the dominant order.

\begin{proposition}[Resonance sets]
For $m=1$ the resonance polynomial is proportional to
\begin{equation}
 (r+1)(r-4)(r-6),
\end{equation}
and the resonances are $r=-1,4,6$.  For $m=2$ the resonance polynomial is proportional to
\begin{equation}
 (r+1)(r-3)(r-4),
\end{equation}
and the resonances are $r=-1,3,4$.
\end{proposition}

\begin{proof}
For $m=1$, $p=-2$.  Using \eqref{eq:u0kdv}, the coefficient of $u_r$ factors as
\[
 \kappa(r-4)\bigl[(r-2)(r-3)-12\bigr]
 =\kappa(r+1)(r-4)(r-6).
\]
For $m=2$, $p=-1$, and \eqref{eq:u0mkdv} gives
\[
 \kappa(r-3)\bigl[(r-1)(r-2)-6\bigr]
 =\kappa(r+1)(r-3)(r-4).
\]
\end{proof}

The universal resonance $r=-1$ corresponds to the arbitrariness of the singular manifold.  Compatibility at the positive resonances is guaranteed in the $m=1,2$ sectors by the transformations of the next section: the equations are point-equivalent to KdV or mKdV, both standard WTC examples \citep{weiss1983,weiss1983b}.

\begin{corollary}[Strong WTC obstruction for higher polynomial degree]
Every polynomial degree $m\ge3$ fails the strong WTC pole criterion at dominant order because $p=-2/m$ is noninteger.  In particular, the first genuinely cubic characteristic speed ($m=3$) produces $p=-2/3$ and a movable algebraic branch behavior rather than a Laurent pole branch.
\end{corollary}

\begin{remark}
Failure of the WTC test is not used here as a logically sufficient proof of nonintegrability in every possible sense.  It is a sharp obstruction to the strong pole-type Painlev\'e property.  The positive integrability statement for $m=1,2$ comes instead from explicit equivalence to the classical integrable equations.
\end{remark}

\section{Affine normal forms of the integrable sectors}
\subsection{Linear characteristic speed: KdV sector}
Let
\[
 P_1(u)=a_0+a_1u,\qquad a_1\ne0.
\]
With the Galilean coordinate $X=x-a_0t$ and $u(x,t)=v(X,t)$, equation \eqref{eq:pgkdv} becomes
\begin{equation}
 v_t+a_1vv_X+\kappa v_{XXX}=0.
 \label{eq:kdvcoeff}
\end{equation}
A nonzero scaling of $X,t,v$ reduces \eqref{eq:kdvcoeff} to the standard KdV equation.  Hence it inherits the inverse-scattering, Lax, recursion-operator and bi-Hamiltonian structures of KdV \citep{gardner1967,miura1968,ablowitz1981,drazin1989}.

\subsection{Quadratic characteristic speed: mKdV/Gardner sector}
Let
\[
 P_2(u)=a_0+a_1u+a_2u^2,\qquad a_2\ne0.
\]
Define
\begin{equation}
 v=u+\frac{a_1}{2a_2},\qquad
 c_0=a_0-\frac{a_1^2}{4a_2}.
 \label{eq:shift}
\end{equation}
Then
\[
 P_2(u)=a_2v^2+c_0.
\]
Using $X=x-c_0t$ gives
\begin{equation}
 v_t+a_2v^2v_X+\kappa v_{XXX}=0,
 \label{eq:mkdvcoeff}
\end{equation}
which is a scaled mKdV equation.  Thus every quadratic polynomial characteristic speed is not merely Painlev\'e-admissible: it lies in the mKdV integrable class after an affine field shift and a Galilean transformation.  The commonly named Gardner form is therefore an affine representation of this same integrable sector.

\begin{theorem}[Polynomial Painlev\'e-integrable sectors]
Within the nonlinear polynomial family \eqref{eq:pgkdv}--\eqref{eq:poly}, the strong WTC pole criterion permits only $m=1,2$.  Both permitted sectors are completely integrable because they are point-equivalent to KdV or mKdV, respectively.
\end{theorem}

\section{Traveling-wave Hamiltonian and algebraic curve}
Set
\begin{equation}
 u(x,t)=U(z),\qquad z=x-st,
 \label{eq:tw}
\end{equation}
where $s$ is the wave speed.  Substitution into \eqref{eq:pgkdv} gives
\[
 \bigl(P_m(U)-s\bigr)U'+\kappa U'''=0.
\]
After one integration,
\begin{equation}
 \kappa U''+A(U)-sU=C,
 \label{eq:tw2}
\end{equation}
where $A'=P_m$ and $C$ is constant.  Multiplication by $U'$ gives the first integral
\begin{equation}
 \frac{\kappa}{2}(U')^2+W_s(U)=E,
 \label{eq:energy}
\end{equation}
where
\begin{equation}
 W_s(U)=\sum_{j=0}^{m}\frac{a_j}{(j+1)(j+2)}U^{j+2}
 -\frac{s}{2}U^2-CU.
 \label{eq:potential}
\end{equation}
Equivalently,
\begin{equation}
 (U')^2=R_{m+2}(U):=\frac{2}{\kappa}\bigl(E-W_s(U)\bigr).
 \label{eq:curve}
\end{equation}
For generic coefficients $R_{m+2}$ has degree $m+2$ and is squarefree.

\begin{theorem}[Genus threshold for traveling waves]
Assume that $R_{m+2}$ is squarefree and of degree $m+2$.  The smooth projective model of
\begin{equation}
 Y^2=R_{m+2}(U)
 \label{eq:hypercurve}
\end{equation}
has genus
\begin{equation}
 g=\left\lfloor\frac{m+1}{2}\right\rfloor.
 \label{eq:genus}
\end{equation}
Consequently, $m=1,2$ give genus-one (elliptic) traveling-wave curves, while every $m\ge3$ gives generic genus $g\ge2$.
\end{theorem}

\begin{proof}
A nonsingular hyperelliptic curve $Y^2=R_d(U)$ with squarefree polynomial $R_d$ of degree $d\ge3$ has genus $\lfloor(d-1)/2\rfloor$ \citep{farkas1992}.  Taking $d=m+2$ yields \eqref{eq:genus}.
\end{proof}

Combining Theorems 3.1 and 5.1 gives the central correspondence.

\begin{corollary}[Coincidence of Painlev\'e and elliptic thresholds]
For nonlinear polynomial characteristic speeds, the degrees that admit an integer WTC principal pole exponent are exactly the degrees whose generic traveling-wave curves are elliptic:
\[
 m\in\{1,2\}
 \quad\Longleftrightarrow\quad
 p=-2/m\in\mathbb Z_{<0}
 \quad\Longleftrightarrow\quad
 g=1.
\]
The first failure, $m=3$, simultaneously gives $p=-2/3$ and generic genus $g=2$.
\end{corollary}

\begin{remark}
For $m\ge3$, repeated roots of $R_{m+2}$ can lower the geometric genus.  Therefore special elliptic or elementary traveling waves may still occur on lower-dimensional parameter sets even though the generic curve is hyperelliptic and the PDE fails the strong WTC criterion.  The distinction between generic structure and degenerate exact solutions is essential.
\end{remark}

\section{Exact Weierstrass reduction in the elliptic sectors}
For $m=1,2$, equation \eqref{eq:curve} has degree at most four.  We now give a representation that does not require first factoring the quartic completely.

\begin{theorem}[Root-based Weierstrass formula]
Let
\begin{equation}
 (U')^2=R_4(U)
 \label{eq:quartic}
\end{equation}
with $R_4$ a quartic polynomial, and let $U_0$ be a simple root of $R_4$.  Define
\begin{align}
 g_2&=\frac{R_4''(U_0)^2}{48}
      -\frac{R_4'(U_0)R_4'''(U_0)}{24},
 \label{eq:g2}\\
 g_3&=-\frac{R_4''(U_0)^3}{1728}
      +\frac{R_4'(U_0)R_4''(U_0)R_4'''(U_0)}{576}
      -\frac{R_4'(U_0)^2R_4''''(U_0)}{384}.
 \label{eq:g3}
\end{align}
Then every nonconstant local solution represented from that simple root can be written
\begin{equation}
 U(z)=U_0+
 \frac{R_4'(U_0)}{4\left[\wpF(z-z_0;g_2,g_3)-R_4''(U_0)/24\right]},
 \label{eq:weier}
\end{equation}
where $z_0$ is a phase constant.
\end{theorem}

\begin{proof}
Put $U=U_0+1/y$.  Since $R_4(U_0)=0$, multiplying \eqref{eq:quartic} by $y^4$ yields
\begin{equation}
 (y')^2=Ay^3+By^2+Cy+D,
 \label{eq:cubicY}
\end{equation}
where
\[
 A=R_4'(U_0),\quad
 B=\frac{R_4''(U_0)}{2},\quad
 C=\frac{R_4'''(U_0)}{6},\quad
 D=\frac{R_4''''(U_0)}{24}.
\]
The affine change
\[
 y=\frac{4w-B/3}{A}
\]
reduces \eqref{eq:cubicY} to
\[
 (w')^2=4w^3-g_2w-g_3,
\]
with \eqref{eq:g2}--\eqref{eq:g3}.  Hence $w=\wpF(z-z_0;g_2,g_3)$, and solving backward for $U$ gives \eqref{eq:weier}.
\end{proof}

The elliptic discriminant is
\begin{equation}
 \Delta=g_2^3-27g_3^2.
 \label{eq:disc}
\end{equation}
When $\Delta\ne0$, the Weierstrass cubic is nonsingular.  When $\Delta=0$, the elliptic function degenerates to trigonometric, hyperbolic or rational functions, in agreement with repeated-root degenerations of the traveling-wave polynomial \citep{whittaker1927,lawden1989,byrd1971}.

\subsection{Real-orbit classification}
The real phase portrait follows directly from \eqref{eq:quartic}.  If $\alpha<\beta$ are consecutive simple real zeros with $R_4(U)>0$ on $(\alpha,\beta)$, then the corresponding orbit is periodic and has period
\begin{equation}
 T=2\int_{\alpha}^{\beta}\frac{\dd U}{\sqrt{R_4(U)}}.
 \label{eq:period}
\end{equation}
If one turning point coalesces into a double root, the period diverges and the periodic orbit approaches a homoclinic solitary wave.  Two suitable double roots generate a heteroclinic kink/antikink limit, while higher multiplicities produce rational degenerations.  These statements are consequences of the local integrability of $\dd U/\sqrt{R_4(U)}$ at simple roots and its logarithmic divergence at a double root.

\section{Example: periodic and solitary mKdV waves}
Consider the normalized focusing mKdV equation
\begin{equation}
 u_t+6u^2u_x+u_{xxx}=0.
 \label{eq:fmkdv}
\end{equation}
For a traveling wave of speed $s$ and zero first integration constant, \eqref{eq:tw2} becomes
\begin{equation}
 U''-sU+2U^3=0.
 \label{eq:mkdvODE}
\end{equation}
The energy law is
\begin{equation}
 (U')^2=sU^2-U^4+2E.
 \label{eq:mkdvenergy}
\end{equation}

\subsection{A nonsingular elliptic orbit}
Take $s=2$ and $E=-0.32=-8/25$.  Then
\begin{equation}
 (U')^2=-U^4+2U^2-\frac{16}{25}
 =\left(\frac85-U^2\right)\left(U^2-\frac25\right).
 \label{eq:examplequartic}
\end{equation}
The positive periodic orbit oscillates between
\[
 U_{\min}=\sqrt{\frac25},\qquad
 U_{\max}=\sqrt{\frac85}.
\]
It has the Jacobi representation
\begin{equation}
 U(z)=\sqrt{\frac85}\,
 \operatorname{dn}\left(\sqrt{\frac85}\,z\,\middle|\,\frac34\right),
 \label{eq:dn}
\end{equation}
with period
\begin{equation}
 T=\frac{2K(3/4)}{\sqrt{8/5}}
 =3.409750627945834\ldots .
 \label{eq:periodnum}
\end{equation}
For the simple root $U_0=\sqrt{8/5}$, formula \eqref{eq:weier} gives
\begin{equation}
 g_2=\frac{73}{75},\qquad
 g_3=\frac{119}{675},\qquad
 \Delta=\frac{1296}{15625}>0.
 \label{eq:invariantsExample}
\end{equation}
Thus the orbit is genuinely elliptic rather than a degenerate elementary wave.

Figure \ref{fig:potential} shows the double-well traveling-wave potential and the turning points at $E=-0.32$.
\begin{figure}[htbp]
\centering
\includegraphics[width=0.78\linewidth]{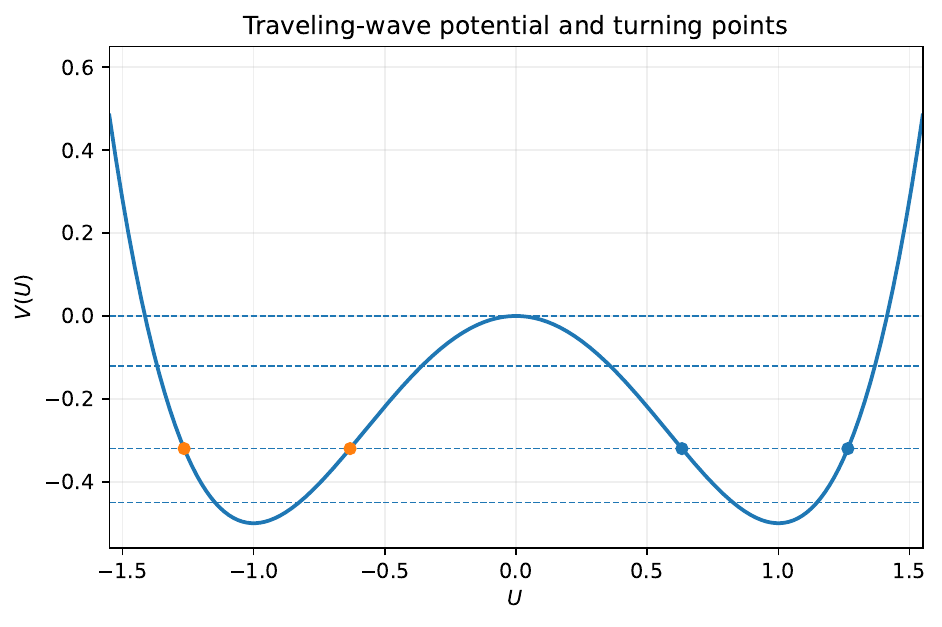}
\caption{Traveling-wave potential $V(U)=-U^2+\tfrac12U^4$ for $s=2$.  The marked turning points correspond to the periodic energy $E=-0.32$.}
\label{fig:potential}
\end{figure}

\subsection{Degenerate solitary-wave limit}
At $E=0$, \eqref{eq:mkdvenergy} becomes
\[
 (U')^2=U^2(s-U^2).
\]
For $s>0$ the double root $U=0$ produces the homoclinic pulse
\begin{equation}
 U(z)=\sqrt{s}\,\operatorname{sech}(\sqrt{s}\,z),
 \label{eq:sech}
\end{equation}
which is the $\Delta\to0$ degeneration of the elliptic family.  This explicitly illustrates the discriminant interpretation in Section 6.

\section{Reproducible numerical verification}
To verify the closed form without relying on symbolic simplification alone, we integrate \eqref{eq:mkdvODE} for $s=2$ with initial data
\[
 U(0)=\sqrt{8/5},\qquad U'(0)=0,
\]
over three periods using an eighth-order Dormand--Prince method at relative tolerance $2\times10^{-12}$ and absolute tolerance $2\times10^{-13}$.  The exact reference is \eqref{eq:dn}.  Figure \ref{fig:compare} shows that the two curves are visually indistinguishable, while Figure \ref{fig:error} resolves the pointwise difference.  The maximum absolute error in the supplied run is
\begin{equation}
 \max |U_{\mathrm{num}}-U_{\mathrm{exact}}|=3.12\times10^{-12}.
 \label{eq:maxerror}
\end{equation}

\begin{figure}[htbp]
\centering
\includegraphics[width=0.78\linewidth]{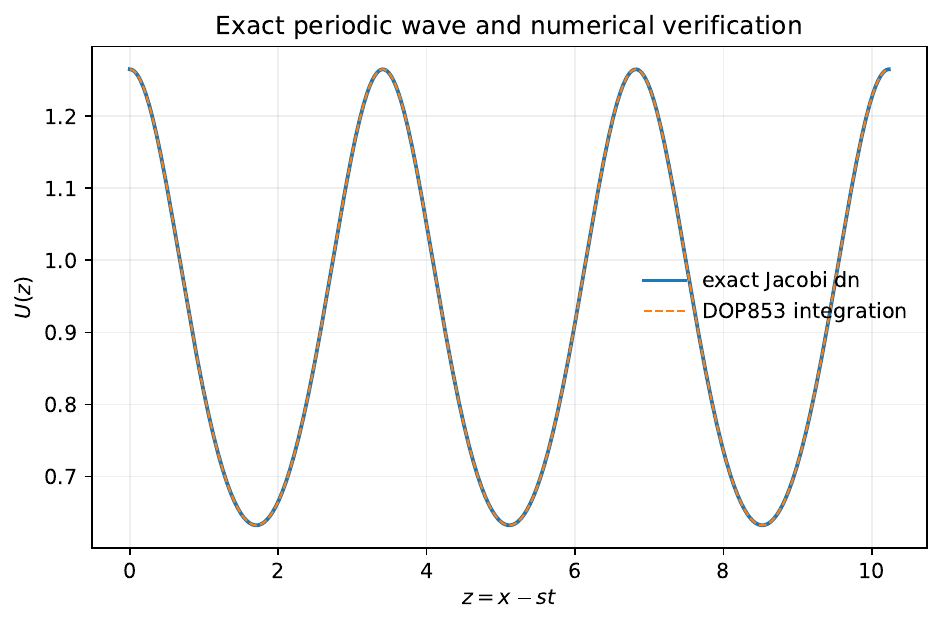}
\caption{Jacobi-$\operatorname{dn}$ solution \eqref{eq:dn} and direct numerical integration of the traveling-wave equation over three periods.}
\label{fig:compare}
\end{figure}

\begin{figure}[htbp]
\centering
\includegraphics[width=0.78\linewidth]{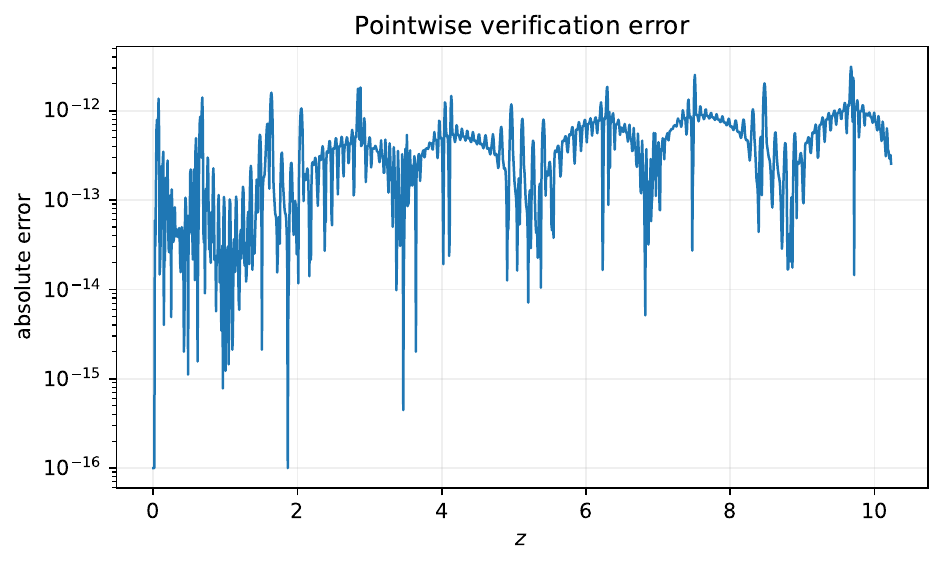}
\caption{Absolute pointwise difference between the exact periodic wave and the numerical IVP solution.}
\label{fig:error}
\end{figure}

The file \texttt{generate\_figures.py} regenerates all three figures and the error value.  The accompanying symbolic audit checks the affine reduction, the first integral, the quartic-to-Weierstrass transformation and the invariants in \eqref{eq:invariantsExample}.

\section{What changes beyond the elliptic threshold?}
For $m=3$, the traveling-wave equation is generically
\[
 (U')^2=R_5(U),
\]
so the quadrature lives on a genus-two hyperelliptic curve.  Simultaneously, the WTC leading exponent is $p=-2/3$.  These two transitions are structurally different manifestations of the same increase in polynomial nonlinearity: the singularity expansion ceases to be a Laurent pole expansion, while the traveling-wave inversion ceases to be generically elliptic.

For higher $m$ the genus increases according to \eqref{eq:genus}, but the WTC obstruction is already present at $m=3$.  One should therefore not infer the integrability of a high-degree generalized KdV equation merely from the existence of one elliptic traveling wave.  Such a wave can arise through a repeated-root degeneration that lowers the genus of a particular energy curve, whereas the PDE-level movable-singularity test probes a much broader local solution structure.

This separation also clarifies the role of Hamiltonian form.  Equation \eqref{eq:ham} exists for all $m$, including $m\ge3$.  Complete integrability requires substantially more structure than one Hamiltonian operator and one family of quadratures.  In the $m=1,2$ sectors that additional structure is supplied by equivalence to KdV/mKdV; beyond the threshold, no such conclusion follows from the present analysis.

\section{Conclusions}
We have organized the polynomial generalized KdV family around three structural diagnostics: WTC singularity balance, Hamiltonian formulation and algebraic-curve genus.  The dominant Painlev\'e exponent is $p=-2/m$, so only the nonlinear degrees $m=1$ and $m=2$ support principal Laurent pole branches.  Their resonance sets coincide with the classical KdV and mKdV patterns, and affine/Galilean transformations establish their complete integrability directly.

The same degrees are precisely those for which a generic traveling-wave first integral defines an elliptic curve.  At $m=3$, the first fractional WTC exponent and the first generic genus-two curve appear together.  This yields a compact threshold principle connecting PDE singularities to the geometry of traveling-wave quadratures.  The Hamiltonian $D_x$ structure, by contrast, survives for every polynomial degree and therefore does not by itself characterize complete integrability.

For the elliptic sectors we derived an explicit root-based Weierstrass formula for a general quartic first integral and linked the Weierstrass discriminant to periodic and degenerate solitary-wave regimes.  The supplied reproducibility files verify the algebraic reduction and a periodic mKdV example independently.

Natural extensions include nonpolynomial characteristic speeds, variable-coefficient deformations that preserve a Painlev\'e branch, and a classification of the exceptional higher-degree parameter manifolds on which the hyperelliptic traveling-wave curve degenerates to genus one.  Such extensions would distinguish accidental elliptic reductions from PDE-level integrability more sharply.

\section*{Declarations}
\textbf{Funding.} No external funding was received for this work.

\textbf{Conflict of interest.} The author declares no conflict of interest.

\textbf{Data availability.} No external datasets were used.  The computational scripts required to reproduce the symbolic checks and numerical figures are included with the manuscript source.

\textbf{Author contributions.} A.H.S.: conceptualization, methodology, formal analysis, software, validation, writing--original draft, writing--review and editing.

\appendix
\section{Resonance-polynomial calculation}
For completeness, we record the algebra behind Proposition 3.2 in the Kruskal gauge $\phi_x=1$.  For $m=1$, set
\[
 u=u_0\phi^{-2}+\epsilon u_r\phi^{r-2}.
\]
At order $\phi^{r-5}$ the linearized dominant terms are
\[
 \kappa(r-2)(r-3)(r-4)u_r
 +a_1u_0(r-4)u_r.
\]
Since $a_1u_0=-12\kappa$, the coefficient factors as
\[
 \kappa(r-4)\{(r-2)(r-3)-12\}
 =\kappa(r+1)(r-4)(r-6).
\]
For $m=2$, with
\[
 u=u_0\phi^{-1}+\epsilon u_r\phi^{r-1},
\]
the coefficient at order $\phi^{r-4}$ is
\[
 \kappa(r-1)(r-2)(r-3)u_r+a_2u_0^2(r-3)u_r.
\]
Using $a_2u_0^2=-6\kappa$ gives
\[
 \kappa(r-3)\{(r-1)(r-2)-6\}
 =\kappa(r+1)(r-3)(r-4).
\]

\section{Quartic-to-Weierstrass algebra}
Let $R=R_4$ and $U=U_0+1/y$, with $R(U_0)=0$.  Taylor expansion is exact because $R$ is quartic:
\[
 R(U_0+y^{-1})
 =\frac{R'}{y}+\frac{R''}{2y^2}+\frac{R'''}{6y^3}+\frac{R''''}{24y^4},
\]
where all derivatives are evaluated at $U_0$.  Since $U'=-y'/y^2$, equation $(U')^2=R(U)$ becomes
\[
 (y')^2=R'y^3+\frac{R''}{2}y^2+\frac{R'''}{6}y+\frac{R''''}{24}.
\]
Writing $A=R'$, $B=R''/2$, $C=R'''/6$, $D=R''''/24$ and
\[
 y=\frac{4w-B/3}{A}
\]
gives
\[
 (w')^2=4w^3+
 \left(\frac{AC}{4}-\frac{B^2}{12}\right)w
 +\frac{A^2D}{16}-\frac{ABC}{48}+\frac{B^3}{216}.
\]
Identifying the last equation with $(w')^2=4w^3-g_2w-g_3$ yields \eqref{eq:g2}--\eqref{eq:g3}.

\end{document}